\documentclass[letterpaper, 10 pt, conference]{ieeeconf}  

\IEEEoverridecommandlockouts                              

\usepackage{amsmath}
\usepackage{amsfonts}
\usepackage{listings}
\usepackage{ntheorem}
\usepackage{newfloat}
\DeclareFloatingEnvironment[placement={!ht},name=List]{mylist}

\newtheorem{assumption}{Assumption}
\newtheorem{remark}{Remark}

\newtheorem{proposition}{Proposition}
\usepackage{listings}
\usepackage{color}
\usepackage{array}
\usepackage{tabu}
\usepackage{mathrsfs,amssymb}
\usepackage{graphicx}
\graphicspath{ {anzcc2018/} }
\usepackage{float}
\usepackage{algorithm}
\usepackage{algpseudocode}
\algdef{SE}[DOWHILE]{Do}{doWhile}{\algorithmicdo}[1]{\algorithmicwhile\ #1}%
\DeclareMathOperator*{\card}{card}
\DeclareMathOperator*{\supp}{supp}
\DeclareMathOperator*{\argmax}{arg\,max}
\title{\LARGE \bf
A Multi-Observer Approach for Attack Detection and Isolation of Discrete-Time Nonlinear Systems
}
	\author{Tianci Yang, Carlos Murguia, Margreta Kuijper, Dragan  Ne\v{s}i\'{c} 
		\thanks{This work was supported by the Australian Research Council under the Discovery Project DP170104099.}
		\thanks{The authors are with the Department of Electrical and Electronics Engineering, the University of Melbourne, Australia.
			{\tt\small tianciy@student.unimelb.edu.au}}
	}

\begin{document}

\maketitle
\thispagestyle{empty}
\pagestyle{empty}

\begin{abstract}
We address the problem of attack detection and isolation for a class of discrete-time nonlinear systems under (potentially unbounded) sensor attacks and measurement noise. We consider the case when a subset of sensors is subject to additive false data injection attacks. Using a bank of observers, each observer leading to an Input-to-State Stable (ISS) estimation error, we propose two algorithms for detecting and isolating sensor attacks. These algorithms make use of the ISS property of the observers to check whether the trajectories of observers are ``consistent'' with the attack-free trajectories of the system. Simulations results are presented to illustrate the performance of the proposed algorithms.
\end{abstract}

\section{Introduction}
%
%
%
%
Traditional control systems composed of interconnected controllers, sensors, and actuators use point-to-point communication architectures. This is no longer suitable when new requirements -- such as modularity, decentralisation of control, integrated diagnostics, quick and easy maintenance, and low cost -- are necessary. To meet these requirements, Networked Control Systems (NCSs) have emerged as a technology that combines control, communication, and computation, and offers the necessary flexibility to meet new demands in distributed and large scale systems.

Recently, security of NCSs has become a very important issue as wireless communication networks might serve as new access points for adversaries trying to disrupt the system dynamics. Cyber-physical attacks on control systems have caused substantial damage to a number of physical processes. One of the most well-known examples is the attack on Maroochy Shire Council's sewage control system in Queensland, Australia that happened in January 2000. The attacker hacked into the controllers that activate and deactivate valves and caused flooding of the grounds of a hotel, a park, and a river with a million liters of sewage. Another incident is the very recent SuxNet virus that targeted Siemens' supervisory control and data acquisition systems which are used in many industrial processes. These incidents show that strategic mechanisms to identify and deal with attacks on NCSs are strongly needed.

In \cite{Fawzi2012}-\nocite{Massoumnia1986}\nocite{Mo2014}\nocite{Vamvoudakis2014}\nocite{Chong2016b}\nocite{Vamvoudakis2012}\nocite{Shoukry2014}\nocite{Yong}\nocite{Park2015}\nocite{Liu2009}\nocite{Teixeira2012b}\nocite{Murguia2016}\nocite{Dolk1}\nocite{Hashemil2017}\nocite{Pasqualetti123}\nocite{Murguia2017d}\nocite{Jairo}\nocite{Carlos_Justin2}\nocite{Sahand2017}\nocite{Carlos_Justin3}\nocite{Carlos_Iman1}\cite{Carlos_Iman2}, a range of topics related to security of control systems have been discussed. In general, they provide analysis tools for quantifying the performance degradation induced by different classes of attacks and propose reaction strategies to counter their effect on the system dynamics. Most of the existing work, however, has considered control systems with linear dynamics, although in most engineering applications the dynamics of the plants being monitored and controlled is highly nonlinear. There are some results addressing the nonlinear case though. In \cite{Kim2016a}, exploiting sensor redundancy, the authors address the problem of sensor attack detection and state estimation for uniformly observable continuous-time nonlinear systems. Similarly, in \cite{Yang2018a}, the problem of state estimation and attack isolation for a class of noisy discrete-time nonlinear system is considered. In particular, the authors propose an observer-based estimator, using a bank of circle-criterion observers, which provides a robust estimate of the system state in spite of sensor attacks and measurement noise, and an estimator-based isolation algorithm without knowing the noise bounds. In this manuscript, we address the problem of attack detection and isolation of a class of discrete-time nonlinear systems in the presence of measurement noise and sensor attacks. We assume that bounds on the measurement noise and an upper bound on the number of attacked sensors are known. We consider the setting when the system has $p$ sensors, all of which are subject to measurement noise and up to $q<p$ of them are attacked. We assume that $q$ is known but the exact subset of sensors being attacked is unknown. Using a bank of the observers, each observer leading to an ISS estimation error, we propose two algorithms for detecting and isolating false data sensor attacks. These algorithms make use of the ISS property of the observers to check whether the trajectories of observers are ``consistent'' with the attack-free trajectories of the system. The main idea behind our algorithms is the following. Each observer in the bank is driven by a different subset of sensors. Thus, without attacks, the observers produce ISS estimation errors with respect to measurement noise only. For every pair of observers in the bank, we compute the largest difference between their estimates. If a pair of observers is driven by a subset of attack-free sensors, then the largest difference between their estimates is also ISS with respect to measurement noise only. However, if there are attacks on some of the sensors, the observers driven by those sensors might produce larger differences than the attack-free ones. These ideas work well under the assumption that less than $p/2$ sensors are attacked, i.e, $q < p/2$.\\[1mm]
\textit{Notation.}\\
We denote the set of real numbers by $\mathbb{R}$, the set of natural numbers by $\mathbb{N}$ , the set of integers by $\mathbb{Z}$, and $\mathbb{R}^{n\times m}$ the set of $n\times m$ matrices for any $m,n \in \mathbb{N}$. For any vector $v\in\mathbb{R}^{n}$, {$v_{J}$} denotes the stacking of all $v_{i}$, $i\in J$ and $J\subset \left\lbrace 1,\ldots,n\right\rbrace$, $|v|=\sqrt{v^{\top} v}$ and $\supp(v)=\left\lbrace i\in\left\lbrace 1,\ldots,n\right\rbrace |v_{i}\neq0\right\rbrace $. For a sequence of vectors $\left\lbrace v(k)\right\rbrace _{k=0}^{\infty}$, we denote by $v_{[0,k]}$ the sequence of vectors $v(i)$, $i=0,\ldots,k$,  $||v||_{\infty}\triangleq \sup_{k\geq 0}|v(k)|$ and $||v||_{T}\triangleq \sup_{0\leq k\leq T}|v(k)|$. We say a sequence $\left\lbrace v(k)\right\rbrace \in l_{\infty}$ if $||v||_{\infty}<\infty$. The binomial coefficient is denoted as $\binom{a}{b}$, where $a,b$ are nonnegative integers. We denote the cardinality of a set $S$ as $\card(S)$.
We denote a variable $m$ uniformly distributed in the interval $(a,b)$ as $m\sim\mathcal{U}(a,b)$.\\

\section{Detection and isolation of sensor attacks}
In this section, we consider a class of discrete-time nonlinear systems subject to sensor attacks and measurement noise. This class of systems has been considered in \cite{Ibrir2007}-\nocite{Sundaram2016}\nocite{Zemouche2006}\cite{Abbaszadeh2008a} in the attack-free case. Consider the system:
\begin{eqnarray}
x^{+}&=&Ax+Gf(Hx)+\rho(u,y)\label{3},\\
\tilde{y}&=&\tilde{C}x+a+\tilde{m},\label{16}
\end{eqnarray}
with state $x\in\mathbb{R}^{n}$, sensor measurement $\tilde{y}\in\mathbb{R}^{p}$, measurement noise $\tilde{m}\in\mathbb{R}^{p}$ satisfying $\left\lbrace \tilde{m}(k)\right\rbrace \in l_{\infty}$, matrices $A \in \mathbb{R}^{n\times n}$, $G \in \mathbb{R}^{n\times r}$, and $H\in\mathbb{R}^{r\times n}$, and attack vector $a\in\mathbb{R}^{p}$. If sensor $i\in\left\lbrace 1,\ldots,p\right\rbrace$ is not attacked, then the $i$-th component of $a(k)$ satisfies $a_{i}(k)=0$ for all $k \geq 0$. Otherwise, sensor $i$ is attacked and $a_{i}(k)$ is arbitrary and possibly unbounded. We denote $W \subseteq\left\lbrace 1,\ldots, p\right\rbrace $ the set of attacked sensors and thus $\supp(a(k)) = W$ for all $k\geq 0$. We assume the set $W$ is fixed and unknown to us. 
The term $\rho(u,y)$ is a known arbitrary real-valued vector that depends on the system inputs and outputs. The state-dependent nonlinearity $f(Hx)$ is an $r$-dimensional vector where each entry is a function of a linear combination of the states:
\begin{equation}\label{230}
f_{i}=f_{i}\left( \sum_{j=1}^{n}H_{ij}x_{j}\right) ,\quad i=1,\ldots,r,
\end{equation}
with $H_{ij}$ denotes the entries of matrix $H$.

Let $q$ be the largest integer such that for each subset of sensors $J\subset\left\lbrace 1,\ldots, p\right\rbrace $ with $\card(J)\geq p-2q>0$ an observer of the form:
\begin{eqnarray}
\begin{split}
\hat{x}_{J}^{+}=&A\hat{x}_{J}+Gf(H\hat{x}_{J}+K_{J}(\tilde{C}_{J}\hat{x}_{J}-\tilde{y}_{J}))\\
&+L_{J}(\tilde{C}_{J}\hat{x}_{J}-\tilde{y}_{J})+\rho(u,y),\label{88}
\end{split}
\end{eqnarray}
exists for $\tilde{y}_{J}\in\mathbb{R}^{\card(J)}$. Here, $\hat{x}_{J}\in\mathbb{R}^{n}$ denotes the estimate of $x$ from $\tilde{y}_{J}$, and $K_{J}\in\mathbb{R}^{r\times\card(J)}$ and $L_{J}\in\mathbb{R}^{n\times\card(J)}$ are the corresponding observer matrices. The matrix $\tilde{C}_{J}$ is the stacking of all $\tilde{C}_{i}$, $i\in J$, where $\tilde{C}_{i}$ denotes the $i$-th row of $\tilde{C}$. Define the estimation error $e_{J}(k):=\hat{x}_{J}(k)-x(k)$. We assume the following.
\vspace{1mm}
\begin{assumption}\label{dd}
	If $a_{J}(k)=0$, there exist constants $c_{J}>0$, $\lambda_{J}\in(0,1)$, and $\gamma_{J}\geq0$ satisfying:
	\begin{equation}
	|e_{J}(k)|\leq c_{J}\lambda^{k}_{J}|e_{J}(0)|+\gamma_{J}||\tilde{m}_{J}||_{k},\label{233}
	\end{equation}
for $k \geq 0$, $e_{J}(0)\in\mathbb{R}^{n}$, and $\tilde{m}_{J}\in\mathbb{R}^{\card(J)}$, $\left\lbrace \tilde{m}_{J}(k)\right\rbrace \in l_{\infty}$.
\end{assumption}

\begin{remark}
 In this manuscript, we consider systems of the form \eqref{3}-\eqref{16} because under certain conditions on $f(\cdot)$, there exist tools -- based on the circle-criterion -- to construct observers of the form \eqref{88} satisfying Assumption \ref{dd}. In particular, we use the result in \cite{Yang2018a}, where the design method is posed as the solution of semidefinite programs.
\end{remark}

\begin{assumption} \label{1001} At most $q$ sensors are attacked, i.e.,
	\begin{equation}
	\mathbf{\card}(W)\leq q,
	\end{equation}
and $q>0$ is a known integer.
\end{assumption}
\vspace{1mm}
\begin{assumption}\label{ll}
The bound on measurement noise is known, i.e.,
	\begin{equation}
		||\tilde{m}||_{\infty}=\bar{m},
	\end{equation}
	and $\bar{m}>0$ is a known constant.
\end{assumption}
\vspace{2mm}

We aim at detecting and isolating sensor attacks on system (\ref{3})-(\ref{16}) for attacks and noise satisfying Assumption \ref{1001} and Assumption \ref{ll}, respectively, and observers of the form \eqref{88} satisfying Assumption \ref{dd}.

\subsection{Detection of sensor attacks}
We construct an observer satisfying Assumption \ref{dd} for system (\ref{3})-(\ref{16}), i.e., considering all sensors, and for each subset $J\subset\left\lbrace 1,\ldots,p\right\rbrace $ of sensors with $\card(J)=p-q$. The obtained estimates are denoted as $\hat{x}$ and $\hat{x}_{J}$, respectively. Define $e=\hat{x}-x$ and let $a=0$; then, under Assumption \ref{dd}, there exist $c>0$, $\lambda\in(0,1)$, and $\gamma\geq0$ such that
$
|e(k)|\leq c\lambda^{k}|e(0)|+\gamma||\tilde{m}||_{k},
$
for all $e(0)\in\mathbb{R}^{n}$, $k\geq 0$, and $\tilde{m}\in\mathbb{R}^{p}$, $\left\lbrace \tilde{m}(k)\right\rbrace \in l_{\infty}$. Because $\lambda\in(0,1)$, it can be easily verified that, for every $\epsilon>0$, there exist $k^{*}$ such that
$
	c\lambda^{k}|e(0)|\leq\epsilon,
$
for all $k\geq k^{*}$, which implies
$
|e(k)|\leq \epsilon+\gamma||\tilde{m}||_{k}\leq\epsilon+\gamma\bar{m},
$
for $k\geq k^{*}$. Also, for each subset $J\subset\left\lbrace 1,\ldots,p\right\rbrace $ with $\card(J)=p-q$, if $a_{J}=0$, there exist $c_{J}>0$, $\lambda_{J}\in(0,1)$, and $\gamma_{J}\geq0$ such that
$
|e_{J}(k)|\leq c_{J}\lambda_{J}^{k}|e(0)|+\gamma_{J}||\tilde{m}_{J}||_{k},
$
for all $e(0)\in\mathbb{R}^{n}$, $k\geq 0$, and $\tilde{m}_{J}\in\mathbb{R}^{p-q}$, $\left\lbrace \tilde{m}_{J}(k)\right\rbrace \in l_{\infty}$. Because $\lambda_{J}\in(0,1)$, there exists $k_{J}^{*}$ such that
$
c_{J}\lambda_{J}^{k}|e(0)|\leq\epsilon,
$
for all $k\geq k_{J}^{*}$, and thus
$
|e_{J}(k)|\leq\epsilon+\gamma_{J}||m_{J}'||_{k}\leq\epsilon+\gamma_{J}\bar{m},
$
for $k\geq k_{J}^{*}$. Let
\[
k^{\star}:=\underset{J\subset\left\lbrace 1,\ldots,p\right\rbrace :\card(J)=p-q}{\max}\left\lbrace k^{*},k_{J}^{*}\right\rbrace,
\]
and define
\begin{equation}\label{ddd}
	\pi(k):=\max_{J\subset\left\lbrace 1,\ldots,p\right\rbrace :\card(J)=p-q}|\hat{x}(k)-\hat{x}_{J}(k)|.
\end{equation}
Let
$J(k)=\argmax_{J\subset\left\lbrace 1,\ldots,p\right\rbrace :\card(J)=p-q}|\hat{x}(k)-\hat{x}_{J}(k)|,
$
 for all $k\geq k^{\star}$. Then, if sensors are attack-free, i.e. $a=0$, we have
\begin{eqnarray}
	\pi(k)&=&|\hat{x}(k)-\hat{x}_{J(k)}(k)|\nonumber\\
	&=&|\hat{x}(k)-x(k)+x(k)-\hat{x}_{J(k)}(k)|\nonumber\\
	&=&|e(k)-e_{J(k)}(k)|\nonumber\\
	&\leq&|e(k)|+|e_{J(k)}(k)|\nonumber\\
	&\leq&2(\epsilon+\bar{\gamma}\bar{m})\label{jj},
\end{eqnarray}
for all $k\geq k^{\star}$, where
\[
\bar{\gamma}:=\underset{J\subset\left\lbrace 1,\ldots,p\right\rbrace :\card(J)=p-q}{\max}\left\lbrace \gamma,\gamma_{J}\right\rbrace.
\]
However, if sensors are under attack, i.e., $a\neq 0$; then, the estimates $\hat{x}(k)$ and $\hat{x}_{J}(k)$ in $\pi(k)$ are likely to be inconsistent and thus lead to larger $\pi(k)$ than the attack-free case. Define
\begin{equation}\label{eee}
	\bar{z} := 2(\epsilon+\bar{\gamma}\bar{m});
\end{equation}
then, $\bar{z}$ can be used as a threshold to detect sensor attacks for $k\geq k^{\star}$. However, it is still possible that for some $k\geq k^{\star}$ and $a_{k}\neq 0$, inequality (\ref{jj}) still holds, which would result in non detection. Then, to improve the detection rate, we perform the detection over windows of $N \in \mathbb{N}$ time-steps. That is, for each $k \in [k^{\star}+(i-1)N,k^{\star}+iN]$, $i \in \mathbb{N}$, we compute $\pi(k)$ and compare it with $\bar{z}$ for every $k$ in the window. If there exists  $k_{1} \in [k^{\star}+(i-1)N,k^{\star}+iN]$, $i \in \mathbb{N}$ such that $\pi(k_{1})>\bar{z}$, then we say that sensors are under attack in the $i$-th window. Otherwise, we say sensors are attack-free in this window. This is formally stated in Algorithm \ref{alg:the_alg2}.

\begin{algorithm}
	\caption{Attack Detection.}
	\label{alg:the_alg2}
	\begin{algorithmic}[1]
		\State \text{Design an observer satisfying Assumption \ref{dd} for system} \text{(\ref{3})-(\ref{16}) and for each subset $ J\subset\left\lbrace 1,\ldots,p\right\rbrace$ with } \text{$\card(J)=p-q$}.
		\State \text{Fix the window size $N\in\mathbb{N}$}.
		\State \text{Calculate} $\bar{z}$ as in (\ref{eee}).
		\State \text{For $i\in\mathbb{Z}_{>0}$, calculate $\pi(k)$ for} \text{$k\in \left[ k^{\star}+(i-1)N, k^{\star}+iN-1\right] $.}
		\State \text{For $ i\in\mathbb{Z}_{>0}$, if $\exists$ $k_{1}\in \left[ k^{\star}+(i-1)N,k^{\star}+iN-1\right] $} \text{such that $\pi(k_{1})>\bar{z}$,
		then sensor attacks occurs in the} \text{$i$-th window, and}
		\begin{equation*}
		detection(i)=1;
		\end{equation*}
		\text{otherwise, sensors are attack-free in the $i$-th window,} \text{and}
		\begin{equation*}
		detection(i)=0.
		\end{equation*}
		\State \text{Return $detection(i)$}
	\end{algorithmic}
\end{algorithm}

Because our knowledge of $||m||_{\infty}$ might be conservative, we consider the case when the actual bound on measurement noise is smaller than $\bar{m}$, i.e., $||m||_{\infty}=\tau \bar{m}$ and $\tau\in(0,1)$. We give a sufficient condition under which sensor attacks cannot be detected by Algorithm \ref{alg:the_alg2} in the $i$-th time window for a given $N>0$.
\begin{proposition}
	Given a time window length $N>0$, if
	\begin{eqnarray}\label{iiii}
	||a||_{k^{\star}+iN-1}\leq(1-\tau)\bar{m};
	\end{eqnarray}
	where $\tau\in(0,1)$, then, $
	\pi(k)\leq\bar{z}
	$ for all $k\in \left[ k^{\star}+(i-1)N, k^{\star}+iN-1\right]$ and sensor attacks cannot be detected by Algorithm \ref{alg:the_alg2} in the $i$-th time window.	
\end{proposition}
\begin{proof}
	For a given time window length $N>0$, sensor attacks cannot be detected by Algorithm \ref{alg:the_alg2} in the $i$-th time window for $a\neq 0$, if we have
	\begin{equation}\label{dddd}
	\pi(k)\leq\bar{z},
 	\end{equation}
	for all $k\in \left[ k^{\star}+(i-1)N, k^{\star}+iN-1\right]$. For $a\neq 0$, we have
    \begin{eqnarray}
    \begin{split}
        	\pi(k)\leq&|e(k)+|e_{J}(k)|\nonumber\\
    \leq&2\epsilon+\gamma(||\tilde{m}||_{ k^{\star}+iN-1}+||a||_{ k^{\star}+iN-1})\nonumber\\
    &+\gamma_{J}(||\tilde{m}||_{k^{\star}+iN-1}+||a_{J}||_{ k^{\star}+iN-1})\nonumber\nonumber\\
    \leq&2(\epsilon+\bar{\gamma}(||\tilde{m}||_{\infty}+||a||_{ k^{\star}+iN-1})\nonumber\\
    \leq&2(\epsilon+\bar{\gamma}(\tau\bar{m}+||a||_{k^{\star}+iN-1})),
    \end{split}
    \end{eqnarray}
    for all $k\in \left[ k^{\star}+(i-1)N, k^{\star}+iN-1\right]$. It follows that the inequality (\ref{dddd}) is satisfied for $a$ satisfying (\ref{iiii}) for all \linebreak $k\in \left[ k^{\star}+(i-1)N, k^{\star}+iN-1\right]$ and thus sensor attacks cannot be detected by Algorithm \ref{alg:the_alg2} in the $i$-th time window.
\end{proof}

Next, we give a sufficient condition under which sensor attacks can always be detected by Algorithm \ref{alg:the_alg2} in the $i$-th time window for a given $N>0$.
\begin{proposition}\label{ss}
	For a given time window length $N>0$, if there exist $k_{1}\in \left[ k^{\star}+(i-1)N,k^{\star}+iN-1\right]$ such that
	\begin{equation}\label{xx}
	|e(k_{1})|> 3(\epsilon+\bar{\gamma}\bar{m});
	\end{equation}
	 then,
	 $
	 \pi(k_{1})>\bar{z}
	 $ and thus sensor attacks can be detected by Algorithm \ref{alg:the_alg2} in the $i$-th time window.
\end{proposition}

\begin{proof}
 For a given time window length $N>0$, sensor attacks can be detected by Algorithm \ref{alg:the_alg2} in the $i$-th time window for $a\neq 0$, if there exist $k_{1}\in \left[ k^{\star}+(i-1)N,k^{\star}+iN-1\right]$ such that
	$
	\pi(k_{1})>\bar{z}
	$. Since there are at most $q$ sensors under attack, we know there exist at least one $\bar{I}\subset\left\lbrace 1,\ldots,p\right\rbrace $ with $\card(\bar{I})=p-q$ such that $a_{\bar{I}}=0$, and
	\begin{equation}\label{cc}
		|e_{\bar{I}}(k)|\leq\epsilon+\gamma_{\bar{I}}||\tilde{m}_{\bar{I}}||_{k},
	\end{equation}
	for $k\geq k^{\star}$. From (\ref{ddd}), we know $\pi(k)\geq|e(k)-e_{\bar{I}}(k)|$ for $k\geq k^{\star}$. If (\ref{xx}) holds, then
	\begin{eqnarray}
		\pi(k_{1})&\geq&||e(k_{1})|-|e_{\bar{I}}(k_{1})||\nonumber\\
		&>&3(\epsilon+\bar{\gamma}\bar{m})-\epsilon-\gamma_{\bar{I}}||\tilde{m}_{\bar{I}}||_{k_{1}}\nonumber\\
		&>&2(\epsilon+\bar{\gamma}\bar{m}),\label{kk}
	\end{eqnarray}
	 which implies sensor attacks can be detected by Algorithm \ref{alg:the_alg2} in the $i$-th time window.
\end{proof}

\subsection{Isolation of sensor attacks}
To perform the isolation, we construct an observer satisfying Assumption \ref{dd} for each subset $J\subset\left\lbrace 1,\ldots,p\right\rbrace $ of sensors with $\card(J)=p-q$ and each subset $S\subset\left\lbrace 1,\ldots,p\right\rbrace $ of sensors with $\card(S)=p-2q$. Hence, by Assumption \ref{dd}, for $a_{S}(k)=0$, there exist $c_{S}>0$, $\lambda_{S}\in(0,1)$, and $\gamma_{S}\geq0$ satisfying
 \begin{equation}\label{61}
 |e_{S}(k)|\leq c_{S}\lambda_{S}^{k}|e(0)|+\gamma_{S}||\tilde{m}_{S}||_{k},
 \end{equation}
 for all $e(0)\in\mathbb{R}^{n}$ and $k\geq0$. Note that, because $\lambda_{S}\in(0,1)$, there always exist $k^{*}_{S}$ such that
$
 	c_{S}\lambda_{S}^{k}|e(0)|\leq\epsilon,
$
 for any $\epsilon>0$ and $k\geq k_{S}^{*}$. Define
 $\bar{k}^{*}:=\max_{J,S}\left\lbrace k_{J}^{*},k_{S}^{*}\right\rbrace.
$
For each subset $J$ with $\card(J)=p-q$, define $\pi_{J}(k)$ as
\begin{equation}\label{efg}
\pi_{J}(k) :=\max_{S\subset J:\card(S)=p-2q}|\hat{x}_{J}(k)-\hat{x}_{S}(k)|.
\end{equation}
Since there are at most $q$ sensors under attack, we know there exist at least one $\bar{I}\subset\left\lbrace 1,\ldots,p\right\rbrace $ with $\card(\bar{I})=p-q$ such that $a_{\bar{I}}=0$ and (\ref{cc}) is satisfied. Define
\begin{eqnarray}
\pi_{\bar{I}}(k)&:=&\underset{S\subset\bar{I}}{\max}|\hat{x}_{\bar{I}}(k)-\hat{x}_{S}(k)|\nonumber\\
&=&\underset{S\subset\bar{I}}{\max}|\hat{x}_{\bar{I}}(k)-x(k)+x(k)-\hat{x}_{S}(k)|\nonumber\\
&\leq&  |e_{\bar{I}}(k)|+\underset{S\subset\bar{I}}{\max}|e_{S}(k)|.
\end{eqnarray}
From (\ref{cc}) and (\ref{61}), we obtain
$
	\pi_{\bar{I}}(k)\leq 2(\epsilon+\gamma_{\bar{I}}'||\tilde{m}_{\bar{I}}||_{k}),
$
for all $k\geq\bar{k}^{*}$, where
\[
\gamma_{\bar{I}}':=\underset{S\subset\bar{I}:\card(S)=p-2q}{\max}\left\lbrace \gamma_{\bar{I}}, \gamma_{S}\right\rbrace.
\]
However, if the subset $J$ of sensors is under attack, i.e., $a_{J}\neq 0$, then $\hat{x}_{J}(k)$ and $\hat{x}_{S}(k)$ in $\pi_{J}(k)$ are more inconsistent and might produce larger $\pi_{J}(k)$. Define
\begin{equation}\label{gg}
	\bar{z}_{J}=2(\epsilon+\gamma_{J}'\bar{m}),
\end{equation}
for each $J\subset\left\lbrace 1,\ldots,p\right\rbrace $ with $\card(J)=p-q$, where
\[
\gamma_{J}':=\underset{S\subset J:\card(S)=p-2q}{\max}\left\lbrace \gamma_{J}, \gamma_{S}\right\rbrace;
\]
then, $\bar{z}_{J}$ can be used as a threshold to isolate attacked sensors. For all $k\geq\bar{k}^{*}$, we select out all the subsets $J\subset\left\lbrace 1,\ldots,p\right\rbrace $ with $\card(J)=p-q$ that satisfy
\begin{equation}\label{777}
\pi_{J}(k)\leq\bar{z}_{J}.
\end{equation}
Denote as $\bar{W}(k)$ the set of sensors that we regard as attack-free at time $k$. Then, $\bar{W}(k)$ is given as the union of all subsets $J$ such that (\ref{777}) holds:
\begin{equation}\label{ccc}
\bar{W}(k):=\underset{J\subset\left\lbrace 1,\ldots,p\right\rbrace :\card(J)=p-q,\pi_{J}(k)\leq\bar{z}_{J}}{\bigcup} J.
\end{equation}
Thus, the set $\left\lbrace 1,\ldots,p\right\rbrace \setminus\bar{W}(k)$ is isolated as the set of attacked sensors at time $k$. However, note that it is still possible that for some $k\geq\bar{k}^{*}$ and some $J\subset \left\lbrace 1,\ldots, p\right\rbrace $ with $\card(J)=p-q$, $a_{J}(k)\neq 0$ but (\ref{777}) still holds. This implies that $J\subset\bar{W}(k)$ even if $a_{J}\neq 0$ and would result in wrong isolation. Therefore, we perform the isolation over windows of $N \in \mathbb{N}$ time-steps. That is, for each $k \in [\bar{k}^{*}+(i-1)N,\bar{k}^{*}+iN]$, $i \in \mathbb{N}$, we compute and collect $\bar{W}(k)$ for every $k$ in the window and select the subset $J$ with $\card(J)\geq p-q$ that is equal to $\bar{W}(k)$ most often in the $i$-th window. We denote this $J$ as $J(i)$. Then, we select $\left\lbrace 1,\ldots,p\right\rbrace \setminus J(i)$ as the set of sensors under attack in the $i$-th window. This is formally stated in Algorithm \ref{alg:alg_3}.

\begin{algorithm}
	\caption{Attack Isolation.}
	\label{alg:alg_3}
	\begin{algorithmic}[1]
		\State \text{Design an observer satisfying Assumption \ref{dd} for each} \text{subset $J\subset\left\lbrace 1,\ldots,p\right\rbrace $ with $\card(J)=p-q$ and each} \text{subset $S\subset\left\lbrace 1,\ldots,p\right\rbrace $ with $\card(S)=p-2q$ }.
		\State \text{Intialize the counter variable $n_{J}(i)=0$ for all $J$} \text{with $\card(J)\geq p-q$ and all $i\in\mathbb{Z}_{>0}$}.
		\State \text{Calculate $\bar{z}_{J}$ for each $J$ with $\card(J)=p-q$ as (\ref{gg}).}
		\State \text{For $i\in\mathbb{Z}_{>0}$ and $\forall k\in \left[ \bar{k}^{*}+(i-1)N,\bar{k}^{*}+iN-1\right] $,} \text{calculate $\pi_{J}(k)$, $\forall J$ with $\card(J)=p-q$ as follows:}
		\begin{equation*}
		\pi_{J}(k)=\max_{S\subset J:\card(S)=p-2q}|\hat{x}_{J}(k)-\hat{x}_{S}(k)|.
		\end{equation*}
		\State  \text{For all $k\in \left[ \bar{k}^{*}+(i-1)N,\bar{k}^{*}+iN-1\right] $, take the} \text{union of all the subsets $J$ such that $\pi_{J}(k)\leq\bar{z}_{J}$:}
		\begin{equation*}
		\bar{W}(k)=\underset{J\subset\left\lbrace 1,\ldots,p\right\rbrace :\card(J)=p-q,\pi_{J}(k)\leq\bar{z}_{J}}{\bigcup} J,
		\end{equation*}
		\State \text{For all $k\in \left[ \bar{k}^{*}+(i-1)N,\bar{k}^{*}+iN-1\right] $, if} \text{$\bar{W}(k)=J$ for some $J$ with $\card(J)\geq p-q$, then} \text{update its corresponding counter variable as follows:}
		\begin{equation*}
		n_{J}(i)=n_{J}(i)+1.
		\end{equation*}
		\State \text{For all $i\in\mathbb{Z}_{>0}$, select the subset $J$ with $\card(J)\geq$} \text{$p-q$ that is equal to $\bar{W}(k)$ most often, i.e.,}
		\begin{equation*}
		J(i)=\underset{J\in\left\lbrace 1,\ldots,p\right\rbrace :\card(J)\geq p-q}{\argmax} n_{J}(i).
		\end{equation*}
		\State \text{For all $i\in\mathbb{Z}_{>0}$, the set of sensors potentially under} \text{attack is given as:}
		\begin{equation*}
		\tilde{A}(i) = \left\lbrace 1,\ldots,p\right\rbrace \setminus J(i).
		\end{equation*}
		\State \text{For all $i\in\mathbb{Z}_{>0}$, return $\tilde{A}(i)$.}
	\end{algorithmic}
\end{algorithm}
Next, we give a sufficient condition under which none of the attacked sensors can be isolated by Algorithm \ref{alg:alg_3} in the $i$-th time window for a given $N>0$ when $||m||_{\infty}=\tau\cdot\bar{m}$ where $\tau\in(0,1)$.
\begin{proposition}
	Given a time window length $N>0$, if
	\begin{equation}\label{ggg}
		||a||_{\bar{k}^{*}+iN-1}\leq (1-\tau)\bar{m};
	\end{equation}
		where $\tau\in(0,1)$, then, for all $J\subset\left\lbrace 1,\ldots,p\right\rbrace $ with $\card(J)=p-q$ and $a_{J}\neq 0$,
	$
	\pi_{J}(k)\leq\bar{z}_{J}$
	for all $k\in \left[ \bar{k}^{*}+(i-1)N,\bar{k}^{*}+iN-1\right] $ and none of attacked sensors can be isolated by Algorithm \ref{alg:alg_3} in the $i$-th time window.	
\end{proposition}
\newpage

\begin{proof}
	For given time window length $N>0$, none of attacked sensors can be isolated by Algorithm \ref{alg:alg_3} in the $i$-th window if $\forall J\subset\left\lbrace 1,\ldots,p\right\rbrace $ with $\card(J)=p-q$ and $a_{J}\neq 0$, we have
	$
	\pi_{J}(k)\leq\bar{z}_{J}$
	for all $k\in \left[ \bar{k}^{*}+(i-1)N,\bar{k}^{*}+iN-1\right] $. For all $k\in \left[ \bar{k}^{*}+(i-1)N,\bar{k}^{*}+iN-1\right] $ and $a_{J}\neq 0$, we have
	\begin{eqnarray}
	\begin{split}
		\pi_{J}(k)\leq&|e_{J}(k)+|e_{S}(k)|\nonumber\\[1mm]
	\leq&2\epsilon+\gamma_{J}(||\tilde{m}_{J}||_{\infty}+||a_{J}||_{\infty})\\[1mm]
	&+\gamma_{S}(||\tilde{m}_{S}||_{\bar{k}^{*}+iN-1}+||a_{S}||_{\bar{k}^{*}+iN-1})\nonumber\\[1mm]
	\leq&2\epsilon+2\gamma_{J}'(||\tilde{m}_{J}||_{\bar{k}^{*}+iN-1}+||a_{J}||_{\bar{k}^{*}+iN-1})\nonumber\\[1mm]
	\leq&2(\epsilon+\gamma'_{J}(\tau\bar{m}+||a||_{\bar{k}^{*}+iN-1})).
	\end{split}
	\end{eqnarray}
	If (\ref{ggg}) holds, then 	$
	\pi_{J}(k)\leq\bar{z}_{J}
	$ for all $J\subset\left\lbrace 1,\ldots,p\right\rbrace $ with $\card(J)=p-q$ and all $k\in \left[ \bar{k}^{*}+(i-1)N,\bar{k}^{*}+iN-1\right] $. Then, none of attacked sensors can be isolated by Algorithm \ref{alg:alg_3} in the $i$-th time window.
\end{proof}

\vspace{1mm}

Next, we give a sufficient condition under which all of attacked sensors can be isolated by Algorithm \ref{alg:alg_3} in the $i$-th time window for a given time window length $N>0$.

\newpage

\begin{proposition}
	 Given a time window length $N>0$, if for all $J\subset \left\lbrace 1,\ldots,p\right\rbrace $ with $\card(J)=p-q$ and $a_{J}\neq 0$, we have
	\begin{equation}\label{xxx}
	|e_{J}(k)|> 3(\epsilon+\gamma_{J}'\bar{m}),
	\end{equation}
	for at least $N/2$ time-steps in the $i$-th time window, then for all $J\subset \left\lbrace 1,\ldots,p\right\rbrace $ with $\card(J)=p-q$ and $a_{J}\neq 0$, we have
	$
	\pi_{J}(k)>\bar{z}_{J}
	$
	for at least $N/2$ time-steps in the $i$-th time window, and all of attacked sensors can be isolated by Algorithm \ref{alg:alg_3} in the $i$-th time window.
\end{proposition}
\begin{proof}
 Since there are at most $q$ sensors under attack, for each subset $J$ with $\card(J)=p-q$ we know there exist at least one $\bar{S}\subset J$ with $\card(\bar{S})=p-2q$ such that $a_{\bar{S}}=0$, and
	$
	|e_{\bar{S}}|\leq\epsilon+\gamma_{\bar{S}}||\tilde{m}_{\bar{S}}||_{k},
$
	for all $k\in \left[ \bar{k}^{*}+(i-1)N,\bar{k}^{*}+iN-1\right]$. By construction of (\ref{efg}), it is satisfied that
	\begin{align*}
	\pi_{J}(k) &= \underset{S\subset J\card(S)=p-2q}{\max}|\hat{x}_{J}(k)-\hat{x}_{S}(k)| \\
	&\geq |e_{J}(k)-e_{\bar{S}}(k)|.
	\end{align*}
	for all $k\geq \bar{k}^{*}$. If (\ref{xxx}) holds at least $N/2$ time-steps in the $i$-th time window, then from triangle inequality, for all $J\subset \left\lbrace 1,\ldots,p\right\rbrace $ with $\card(J)=p-q$ and $a_{J}\neq 0$, we have
	\begin{eqnarray}
	\pi_{J}(k)&\geq&||e_{J}(k)|-|e_{\bar{S}}(k)||\nonumber\\
	&>&3(\epsilon+\gamma_{J}'\bar{m})-\epsilon-\gamma_{\bar{S}}||\tilde{m}_{\bar{S}}||_{k}\nonumber\\
	&>&2(\epsilon+\gamma_{J}'\bar{m}),
	\end{eqnarray}
	for at least $N/2$ time-steps in the $i$-th time window, which implies all of attacked sensors can be isolated by Algorithm \ref{alg:alg_3} in the $i$-th time window.
\end{proof}
\begin{remark}
The performance of Algorithm \ref{alg:alg_3} can be arbitrarily improved by increasing the length of the time window $N$ at the price of increasing the time needed for isolation.
\end{remark}
\vspace{2mm}
\textbf{Example 1}
 Consider the discrete-time nonlinear system subject to measurement noise and sensor attacks:
\begin{eqnarray}
x^{+}&=&\left[ \begin{matrix}\label{e1}
1&\delta\\
0&1
\end{matrix}\right]x+\left[ \begin{matrix}
\frac{1}{2}\delta\alpha \sin (x_{1}+x_{2})\\
\delta\alpha\sin (x_{1}+x_{2})
\end{matrix}\right]
+\left[ \begin{matrix}
\delta u\\
\delta u
\end{matrix}\right], \\
\tilde{y}&=&\left[ \begin{matrix}\label{e2}
3&0.3\\
3&0.6\\
6&0.9\\
1.2&12
\end{matrix}\right] x+a+\tilde{m}.
\end{eqnarray}
with $\delta=0.1$, $\alpha=1$, and $\tilde{m}_{i}\sim\mathcal{U}(-0.5,0.5)$ for $i\in\left\lbrace 1,2,3,4\right\rbrace $. Using the design method proposed in \cite{Yang2018a}, we find that circle-criterion observers of the form (\ref{88}) satisfying Assumption \ref{dd} exist for each subset $J\subset\left\lbrace 1,2,3,4\right\rbrace $ with $\card(J)\geq 1$. Since $p=4$, by Assumption \ref{1001}, the maximum number of attacks is $q=1$. We design a circle-criterion observer for the whole system and for each $J\subset\left\lbrace 1,2,3,4\right\rbrace $ with $\card(J)=3$ and each $S\subset\left\lbrace 1,2,3,4\right\rbrace $ with $\card(S)=2$. Therefore, in total, $\binom{4}{3} + \binom{4}{2}+1 = 11$ observers are designed. We obtain their ISS gains by montecarlo simulations. Theses eleven observers are initialized at $\hat{x}(0)=x(0)$ and $x_{1}(0),x_{2}(0)$ are randomly selected from a standard normal distribution; thus, $\epsilon=0$. We let $N=50,100,200$ and evaluate Algorithm 1 and Algorithm 2 for $1000$ time-steps. For Algorithm \ref{alg:the_alg2}, we let $W=\left\lbrace 2\right\rbrace $, which means the $2$-nd sensor is under attack, and $a_{2}\sim\mathcal{U}(-c,c)$ with $c$ given by $0.7$ and $1$. We run Algorithm \ref{alg:the_alg2} with $\binom{4}{3}+1=5$ observers. The detection results are shown in Figures \ref{fig:15}-\ref{fig:16}. For Algorithm \ref{alg:alg_3}, we let $W=\left\lbrace 3\right\rbrace $ and $a_{3}\sim\mathcal{U}(-d,d)$ with $d$ given by $2$ and $5$. We run Algorithm \ref{alg:alg_3} with $\binom{4}{3}+\binom{4}{2}=10 $ observers. We say sensor $0$ is under attack in the $i$-th window when $\tilde{A}_{i}=\emptyset$. The isolation results are shown below in Figures \ref{fig:1}-\ref{fig:2}.\\[1mm]
\begin{figure}[h]
	\includegraphics[width=0.5\textwidth]{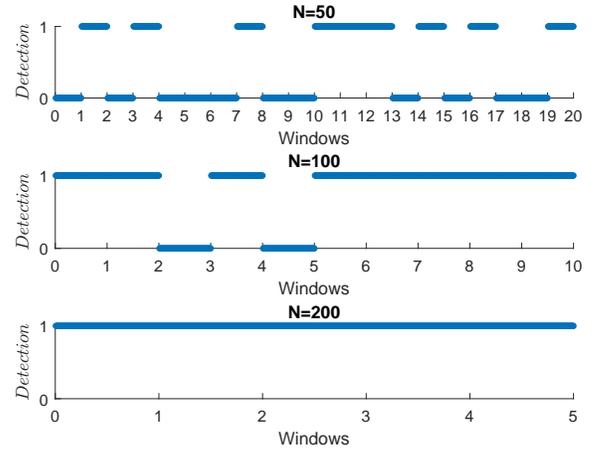}
	\caption{Attack detection, $a_{2}\sim\mathcal{U}(-0.7,0.7)$.}
	\label{fig:15}
	\centering
\end{figure}
\begin{figure}[h]
	\includegraphics[width=0.5\textwidth]{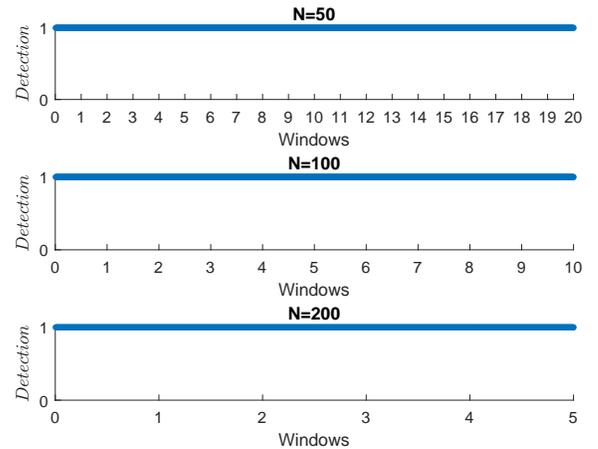}
	\caption{Attack detection, $a_{2}\sim\mathcal{U}(-1,1)$.}
	\label{fig:16}
	\centering
\end{figure}
\begin{figure}[h]
\includegraphics[width=0.5\textwidth]{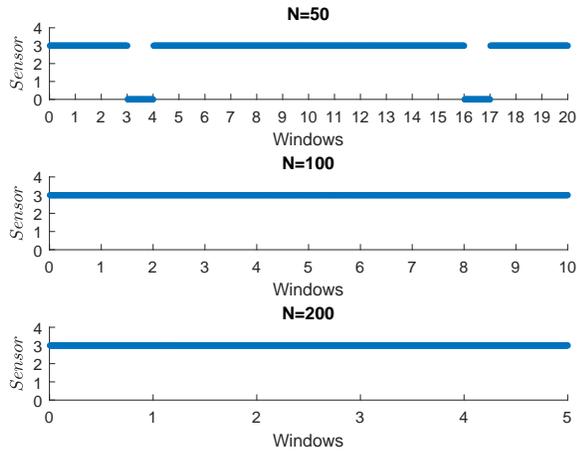}
\caption{Attack isolation, $a_{3}\sim\mathcal{U}(-2,2)$.}
\label{fig:1}
\centering
\end{figure}
\begin{figure}[h]
\includegraphics[width=0.5\textwidth]{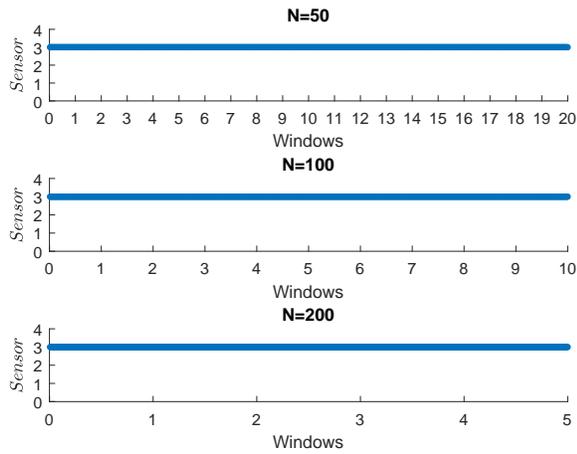}
\caption{Attack isolation, $a_{3}\sim\mathcal{U}(-5,5)$.}
\label{fig:2}
\centering
\end{figure}
\section{Conclusion}

Assuming that a sufficiently small subset of sensors is subject to additive false data injection attacks, we have proposed two algorithms for detecting and isolating sensor attacks for a class of discrete-time nonlinear systems subject to measurement noise using a multi-observer approach. We have provided simulations results to illustrate the performance of the proposed algorithms. The performance of our algorithms can be improved by increasing the length of the time window $N$ at the price of increasing the time needed for detection and isolation.


\bibliographystyle{ieeetr}
\bibliography{Observer} 

%

%
%
%




\end{document}